\documentclass[sort&compress,a4paper,oneside,3p,12pt]{melsarticle}

\usepackage{amssymb}
\usepackage{amsmath,amssymb,amsfonts,graphicx,float,multicol,fleqn}
\usepackage{multirow}

\newtheorem{theorem}{Theorem}

\newtheorem{proof}{Proof}

\begin{document}
\begin{frontmatter}

\title{Numerical approximations for population growth model by Rational Chebyshev and Hermite Functions collocation approach: A comparison}


\author[a]{K. Parand\corref{cor}}
\cortext[cor]{Corresponding author. Tel:+98 21 22431653; Fax:+98 21 22431650.}
\ead{k\_parand@sbu.ac.ir}
\fntext[a]{Member of research group of Scientific Computing.}
\author{A. R. Rezaei}
\ead{alireza.rz@gmail.com}
\author{A. Taghavi}
\ead{amirtaghavims2@yahoo.com}
\address{Department of Computer Sciences, Shahid Beheshti University, G.C., Tehran, Iran}

\begin{abstract}
This paper aims to compare rational Chebyshev (RC) and Hermite functions (HF) collocation approach to solve the Volterra's model for population growth of a species within a closed system. This model is a nonlinear integro-differential equation where the integral term represents the effect of toxin. This approach is based on orthogonal functions which will be defined. The collocation method reduces the solution of this problem to the solution of a system of algebraic equations. We also compare these methods with some other numerical results and show that the present approach is applicable for solving nonlinear integro-differential equations.
\end{abstract}

\begin{keyword}
Collocation method; Spectral method; Volterra's population model; Nonlinear integro-differential equation; Hermite functions; Rational Chebyshev; Numerical approximations; Population growth model

\PACS 02.60.Lj, 02.70.Hm, 87.23.Cc
\end{keyword}

\end{frontmatter}

\section{Introduction}
Many science and engineering problems arise in unbounded domains. Different spectral methods have been proposed for solving problems in unbounded domains. The most common method is through the use of polynomials that are orthogonal over unbounded domains, such as the Hermite spectral method and the Laguerre spectral method \cite{Coulaud,Funaro.Kavian,Funaro.Appl. Numer. Math.1990,Guo.Math. Comp.1999,Guo.num2000,Maday,Shen,Siyyam,Iranzo}.

Guo \cite{Guo.J. Math. Anal. Appl.1998,Guo.com2000,Guo.J. Math. Anal. Appl.2000} proposed a method that proceeds by mapping the original problem in an unbounded domain to a problem in a bounded domain, and then using suitable Jacobi polynomials to approximate the resulting problems.

Another approach is replacing infinite domain with $[-L, L]$ and semi-infinite interval with $[0, L]$ by choosing $L$, sufficiently large.
This method is named domain truncation \cite{BoydBook}.

Another effective direct approach for solving such problems is based on rational approximations.
Christov \cite{Christov.SIAM J. Appl. Math.1982} and Boyd \cite{Boyd.J. Comput. Phys.1987(69),Boyd1987} developed some spectral methods on unbounded intervals by using mutually orthogonal systems of rational functions.
Boyd \cite{Boyd1987} defined a new spectral basis, named rational Chebyshev functions on the semi-infinite interval, by mapping to the Chebyshev polynomials.
Guo et al. \cite{Guo.sci2000} introduced a new set of rational Legendre functions which are mutually orthogonal in $L^2(0,+\infty)$.
They applied a spectral scheme using the rational Legendre functions for solving the Korteweg-de Vries equation on the half line.
Boyd et al. \cite{Boyd2003} applied pseudospectral methods on a semi-infinite interval and compared rational Chebyshev, Laguerre and mapped Fourier sine.

Parand et al. \cite{Parand.Appl. Math. Comput.2004,Parand.Int. J. Comput. Math.2004,Parand.Phys. Scripta2004,Parand.Shahini.Phys.let.A,Parand.CAM,Parand.JCP,Parand.Dehghan.Rezaei.CPC} applied spectral method to solve nonlinear ordinary differential equations on semi-infinite intervals. Their approach was based on rational Tau and collocation method.

Among these, an approach consists of using the collocation method or the pseudospectral method based on the nodes of Gauss formulas related to unbounded intervals \cite{Iranzo}.\\

Collocation method has become increasingly popular for solving differential equations also they are very useful in providing highly accurate solutions to differential equations.\\

We aim to compare rational Chebyshev collocation (RCC) approach and Hermite functions collocation (HF) collocation approach to solve a population growth of a species within a closed system.

 This paper is arranged as follows: in subsection \ref{Volterra.Intro}, the Volterra's population model is taken into consideration and some of the traditional methods that solved it are discussed. In sections \ref{RCF.Intro} and \ref{HF.Intro}, we describe the properties of rational Chebyshev and Hermite functions. In section \ref{Voltera.Solve}, the proposed methods are applied to solve Volterra's population model. This equation is first converted to an equivalent nonlinear ordinary differential equation and then our methods are applied to solve this new equation, and then a comparison is made with existing methods that were reported in the literature. The numerical results and advantages of the methods are discussed in the final section.
\subsection{Volterra's Population Model}\label{Volterra.Intro}
The Volterra's model for population growth of a species within a closed system is given in \cite{Scudo,Small} as
\begin{equation}\label{MainVolterra.}
\frac{dp}{d\tilde{t}}=ap-bp^2-cp\int_0^{\tilde{t}}p(x)dx,\ \ \ \
p(0)=p_0,
\end{equation}
where $a > 0$ is the birth rate coefficient, $b > 0$ is the crowding coefficient, and $c > 0$ is the toxicity coefficient. The coefficient $c$ indicates the essential behavior of the population evolution before its level falls to zero in the long term. $p_{0}$ is the initial population, and $p=p(\tilde{t})$ denotes the population at time $\tilde{t}$. \\
This model is a first-order integro-ordinary differential equation where the term $cp\int_0^{\tilde{t}}p(x)dx$ represents the effect of toxin accumulation on the species. We apply scale time and population by introducing the nondimensional variables $$t=\frac{\tilde{t}c}{b}\quad \mbox{and} \quad  u=\frac{pb}{a}, $$ to obtain the nondimensional problem
\begin{equation}\label{VolterraPopulationInt.}
\kappa \frac{du}{dt}=u-u^2-u\int_0^tu(x)dx,\ \ \ \ u(0)=u_0,
\end{equation}
where $u(t)$ is the scaled population of identical individuals at
time $t$, and $\kappa = c/(ab)$ is a prescribed nondimensional
parameter. The only equilibrium solution of Eq. (\ref{VolterraPopulationInt.}) is the trivial
solution $u(t)=0$ and the analytical solution \cite{Tebeest}
$$u(t)=u_0 \thinspace exp(\frac {1}{\kappa}\int_0^t[1-u(\tau)-
\int_0^{\tau}u(x)dx]d
\tau),$$ shows that $u(t)> 0$ for all $t$ if $u_0 > 0$. \\

The solution of  Eq. (\ref{MainVolterra.}) has been of considerable concern. Although
a closed form solution has been achieved in \cite{Scudo,Small}, it was formally shown that the closed form solution cannot lead to any insight into the behavior of the population evolution \cite{Scudo}. In the literature, several numerical solutions for Volterra's population model have been reported. In \cite{Scudo}, the successive approximations method was suggested for the solution of Eq. (\ref{VolterraPopulationInt.}), but was not implemented. In this case, the solution $u(t)$ has a smaller amplitude compared to the amplitude of $u(t)$ for the case $\kappa \ll 1$.

In \cite{Small}, the singular perturbation method for solving Volterra's
population model is considered. The author scaled out the
parameters of Eq. (\ref{MainVolterra.}) as much as possible and considered four
different ways to do this. He considered two cases $\kappa =
c/(ab)$ small and $\kappa = c/(ab)$ large.

It is shown in \cite{Small} that for the case $\kappa \ll 1$, where populations
are weakly sensitive to toxins, a rapid rise occurs along the
logistic curve that will reach a peak and then is followed by a
slow exponential decay. And, for large $\kappa$, where populations
are strongly sensitive to toxins, the solutions are proportional
to $sech^2(t).$

In \cite{Tebeest}, several numerical algorithms namely Euler method, modified Euler method, classical fourth-order Runge-Kutta method and Runge-Kutta-Fehlberg method for the solution of Eq. (\ref{VolterraPopulationInt.}) are obtained. Moreover, a phase-plane analysis is implemented. In \cite{Tebeest}, the numerical results are correlated to give insight on the problem and its solution without using perturbation techniques. However, the performance of the traditional numerical techniques is well known in that it provides grid points only, and in addition, it requires large amounts of calculations.

In \cite{Al-Khaled} Adomian decomposition method and Sinc-Galerkin method were compared for the solution of some
mathematical population growth models.

In \cite{Wazwaz}, the series solution method and the decomposition method are implemented independently to Eq. (\ref{VolterraPopulationInt.}) and to a related nonlinear ordinary differential equation. Furthermore, the Pad\'{e} approximations are used in the analysis to capture the essential behavior of the populations $u(t)$ of identical individuals and approximation of $u_{max}$ and the exact value of $u_{max}$ for different $\kappa$ were compared.

The authors of \cite{Parand.Appl. Math. Comput.2004,Parand.Int. J. Comput. Math.2004,Parand.Phys. Scripta2004} applied spectral method to solve Volterra's population on a semi-infinite interval. This approach is based on a rational Tau method. They obtained the operational matrices of derivative and the product of rational Chebyshev and Legendre functions and then applied these matrices together with the Tau method to reduce the solution of this problem to the solution of a system of algebraic equations.

In \cite{Parand.Hojati} second derivative multistep methods (denoted SDMM) are used to solve Volterra's model. They first converted the model to a nonlinear ordinary differential equation and then the new SDMM applied to solve this equation.

In \cite{Ramezani.Razzaghi} the approach is based upon composite spectral functions approximations. The properties of composite spectral functions consisting of few terms of orthogonal functions are utilized to reduce the solution of the Volterra's model to the solution of a system of algebraic equations.

In \cite{Marzban.Hoseini} a numerical method based on hybrid function approximations was proposed to solve Volterra's population model. These hybrid functions consist of block-pulse and Lagrange-interpolating polynomials.

Momani et al. \cite{Momani.Qaralleh} and Xu \cite{Xu} used a numerical and Analytical algorithm for approximate solutions of a fractional population growth model respectively. The first algorithm is based on Adomian decomposition method (ADM) with Pad\'{e} approximants and the second algorithm is based on homotopy analysis method (HAM).

In total, in recent years, numerous works have been focusing on the development of more advanced and efficient methods for initial value problems especially for stiff systems.
\section{Rational Chebyshev Functions}\label{RCF.Intro}
This section is devoted to introducing rational Chebyshev functions (which we denote (RC)) and expressing some basic properties of them that will be used to construct the RC collocation (RCC) method.
rational Chebyshev functions denoted by $R_{n}(x)$ are generated from well known
Chebyshev polynomials by using the algebraic mapping $\phi(x) = (x - L)/(x + L)$ \cite{BoydBook,Boyd1987,Parand.Shahini.Phys.let.A,Guo.Shen.Wang2002}
\begin{equation}
R_{n}(x)=T_{n}(\phi(x)),
\end{equation}
where $L$ is a constant parameter and $T_{n}(y)$ is the Chebyshev polynomial of degree $n$.
The constant parameter $L$ sets the length scale of the mapping. Boyd \cite{BoydBook,Boyd1982} offered guidelines for optimizing the map parameter $L$ where $L>0$.
Using properties of Chebyshev polynomials and RC, we have

\begin{align}
\nonumber R_{n}(x)&=\sum_{i=0}^{\lfloor \frac{n}{2}\rfloor}(-1)^{i}2^{n-2i}{{n-i}\choose{i}}\left(\frac{x-L}{x+L}\right)^{n-2i}\\
&=\sum_{i=0}^{\lfloor \frac{n-1}{2}\rfloor}(-1)^{i}2^{n-2i-1}{{n-i-1}\choose{i}}\left(\frac{x-L}{x+L}\right)^{n-2i}
\end{align}
Other properties of RC and a complete discussion on approximating functions by RC are given
in \cite{Parand.Shahini.Phys.let.A,Guo.Shen.Wang2002}.
\subsection{Rational Chebyshev functions approximation}
Let
\begin{equation}
\Re_{N}=span\{R_0,R_1,...,R_N\}.
\end{equation}
We define $P_N:L_{w}^{2}(\Lambda)\rightarrow\Re_{N}$ by
\begin{equation}
P_{N}u(x)=\sum_{k=0}^{N}a_{k}R_{k}(x)
\end{equation}
To obtain the order of convergence of rational Chebyshev approximation,
we define the space
\begin{equation}
H_{w,A}^{r}(\lambda)=\{v:v \text{ is measurable and } \|v\|_{r,x,A}\ <\infty\},
\end{equation}
where the norm is induced by
\begin{equation}
\|v\|_{r,x,A}=\left( \sum_{k=0}^{r}\left\|(x+1)^{\frac{r}{2}+k}\frac{d^{k}}{dx^{k}}v\right\|_{w}^{2} \right)^{\frac{1}{2}},
\end{equation}
and A is the Sturm-Liouville operator as follows:
\begin{equation}
A{v(x)}=-w^{-1}(x)\frac{d}{dx}\left( w^{-1}(x)\frac{d}{dx}v(x)\right).
\end{equation}
$w(x)$ is the weight function and $w(x) = \sqrt{L}/(\sqrt{x}(x +
L))$. We have the following theorem for the convergence:
\begin{theorem}
For any $v\in H_{w,A}^{r}(\Lambda)$ and $r\geq 0$,
\begin{equation}
{\|{P_{N}v-v}\|}_{w}\leq{cN^{-r}\|{v}\|_{r,w,A}}.
\end{equation}
\end{theorem}
\begin{proof}
A complete proof is given by Guo et al.\cite{Guo.Shen.Wang2002}.
\end{proof}
This theorem shows that the rational Chebyshev approximation has exponential
convergence.

\section{Properties of Hermite Functions}\label{HF.Intro}
In this section, we detail the properties of the Hermite functions (HF) that will be used to construct the Hermite functions collocation (HFC) method.
First we note that the Hermite polynomials are generally not suitable in practice due to their wild asymptotic behavior at infinities \cite{ShenWang2008}.\\
Hermite polynomials with large $n$ can be written in direct formula as follow:
\begin{align}\nonumber
H_{n}(x)&\sim\frac{\Gamma(n+1)}{\Gamma(n/2+1)}e^{x^2/2}\cos{(\sqrt{2n+1}x-\frac{n\pi}{2})}\\\nonumber
&\sim n^{n/2}e^{x^2/2}\cos(\sqrt{2n+1}x-\frac{n\pi}{2}).
\end{align}
Hence, we shall consider the so called Hermite functions.
The normalized Hermite functions of degree $n$ is defined by
\begin{align}\nonumber
\widetilde{H}_{n}(x)=\frac{1}{\sqrt{2^n n!}}e^{-x^2/2}H_{n}(x), \quad n\geq 0, x\in \mathbb{R}.
\end{align}
Clearly, \{$\widetilde{H}_{n}$\} is an orthogonal system in $L^2(\mathbb{R})$,i.e.,
\begin{align}\nonumber
\int^{+\infty}_{-\infty}\widetilde{H}_{n}(x)\widetilde{H}_{m}(x)dx=\sqrt{\pi}\delta_{mn}.
\end{align}
where $\delta_{nm}$ is the Kronecker delta function.\\
In contrast to the Hermite polynomials, the Hermite functions are well behaved with the decay property:
\begin{align}\nonumber
|\widetilde{H}_{n}(x)|\longrightarrow 0, \quad \text{as}\quad |x|\longrightarrow \infty,
\end{align}
and the asymptotic formula with large $n$ is
\begin{align}\nonumber
\widetilde{H}_{n}(x)\sim n^{-\frac{1}{4}}\cos(\sqrt{2n+1}x-\frac{n\pi}{2})
\end{align}
The three-term recurrence relation of Hermite polynomials implies
\begin{eqnarray}\nonumber
&\widetilde{H}_{n+1}(x)=x\sqrt{\frac{2}{n+1}}\widetilde{H}_{n}(x)-\sqrt{\frac{n}{n+1}}\widetilde{H}_{n-1}(x),\quad n\geq{1},\\\nonumber
&\noindent\widetilde{H}_{0}(x)=e^{-x^2/2}, \quad \widetilde{H}_{1}(x)=\sqrt{2}xe^{-x^2/2}.
\end{eqnarray}
Using recurrence relation of Hermite polynomials and the above formula leads to
\begin{align}\nonumber
\widetilde{H}'_{n}(x)&=\sqrt{2n}\widetilde{H}_{n-1}(x)-x\widetilde{H}_{n}(x)\\\nonumber
&=\sqrt{\frac{n}{2}}\widetilde{H}_{n-1}(x)-\sqrt{\frac{n+1}{2}}\widetilde{H}_{n+1}(x).
\end{align}
and this implies
\begin{align}\nonumber
\int_{\mathbb{R}}{\widetilde{H}'_n(x)\widetilde{H}'_m(x)dx}=
\begin{cases}
-\frac{\sqrt{n(n-1)\pi}}{2}, & m=n-2, \\
\sqrt{\pi}(n+\frac{1}{2}), & m=n,\\
-\frac{\sqrt{(n+1)(n+2)\pi}}{2}, & m=n+2, \\
0, & \text{otherwise}.
\end{cases}
\end{align}
Let us define
\begin{align}\nonumber
\widetilde{P}_N:=\{u:u=e^{-x^2/2}v,\forall v\in{P_N}\}.
\end{align}
where $P_N$ is the set of all Hermite polynomials of degree at most $N$.\\
We now introduce the Gauss quadrature associated with the Hermite functions approach.\\
Let $\{x_j\}_{j=0}^{N}$ be the Hermite-Gauss nodes and define the weights
\begin{align}\nonumber
\widetilde{w}_j=\frac{\sqrt{\pi}}{(N+1)\widetilde{H}_{N}^{2}(x_j)}, \quad 0\leq j\leq N.
\end{align}
Then we have
\begin{align}\nonumber
\int_{\mathbb{R}}p(x)dx=\sum_{j=0}^{N}p(x_j)\widetilde{w}_j, \quad \forall p\in \widetilde{P}_{2N+1}.
\end{align}
For a more detailed discussion of these early developments, see the \cite{ShenTangHighOrder,ShenTangWang}.
\subsection{Approximations by Hermite Functions}
Let us define $\Lambda:=\{x|-\infty<x<\infty\}$
and
\begin{align}\nonumber
\mathcal{H}_N=span\{\widetilde{H}_0(x),\widetilde{H}_1(x),...,\widetilde{H}_ N(x)\}
\end{align}
The $L^{2}{(\Lambda)}$-orthogonal projection $\tilde{\xi}_N:L^{2}{(\Lambda)}\longrightarrow \mathcal{H}_N$ is a mapping in a way that for any $v\in L^{2}{(\Lambda)}$,
\begin{align}\nonumber
<\tilde{\xi}_N{v}-v,\phi>=0, \quad \forall \phi\in \mathcal{H}_N
\end{align}
or equivalently,
\begin{align}\nonumber
\tilde{\xi}_N{v(x)}=\sum_{l=0}^{N}\tilde{v}_{l}\widetilde{H}_l(x).
\end{align}
To obtain the convergence rate of Hermite functions we define the space $H_{A}^{r}(\Lambda)$ defined by
\begin{align}\nonumber
H^{r}_{A}{(\Lambda)}=\{v|v \text{ is measurable on } \Lambda \text{ and }{\|v\|}_{r,A}<\infty\},\nonumber
\end{align}
and equipped with the norm $\|v\|_{r,A}=\|A^{r}v\|$. For any $r > 0$, the space $H^{r}_{A}{(\Lambda)}$ and its norm are defined by space interpolation. By induction, for any non-negative integer $r$,
\begin{align}\nonumber
A^{r}v(x)=\sum_{k=0}^{r}(x^2+1)^{(r-k)/2}p_k(x){\partial}_{x}^{k}v(x),
\end{align}
where $p_k(x)$ are certain rational functions which are bounded uniformly on $\Lambda$. Thus,
\begin{align}\nonumber
\|v\|_{r,A}\leq c \left(\sum_{k=0}^{r}\parallel(x^2+1)^{(r-k)/2}p_k(x)\partial_{x}^{k}v\parallel \right)^{1/2}.
\end{align}
\begin{theorem}\nonumber
For any $v \in H^{r}_{A}(\Lambda)$, $r\geq1$ and $0\leq \mu \leq r$,
\begin{equation}
\parallel{\tilde{\xi}}_{N}v-v{\parallel_{\mu}\leq cN^{1/3+(\mu-1)/2}\parallel v\parallel}_{r,A}.
\end{equation}
\end{theorem}
{\it Proof}. A complete proof is given by Guo et al. \cite{GuoShenXu2003}. Also same theorem has been proved by Shen et al. \cite{ShenWang2008}.

\section{Solving Volterra's Population Model}\label{Voltera.Solve}
In this section, we study an algorithm for solving Volterra's population model by using the collocation method based on rational Chebyshev and Hermite functions. We first convert
Volterra's population model in Eq. (\ref{VolterraPopulationInt.}) to an equivalent nonlinear
ordinary differential equation. Let
\begin{equation}\label{EqInt.}
y(t)=\int_0^t u(x) dx.
\end{equation}
 This leads to
\begin{align}\label{EqsdiffInt.}
y'(t) = u(t), y''(t) = u'(t).
\end{align}
Inserting  Eq. (\ref{EqInt.}) and Eq. (\ref{EqsdiffInt.}) into Eq. (\ref{VolterraPopulationInt.})
yields the nonlinear differential equation
\begin{equation}\label{Main.voltra}
\kappa y''(t)=y'(t)-(y'(t))^2-y(t)y'(t),
\end{equation}
with the initial conditions
\begin{align}\label{voltraBoundry}
y(0) =&  0,\\ \nonumber
y'(0) =&  u_0,
\end{align}
that were obtained by using  Eq. (\ref{EqInt.}) and Eq. (\ref{EqsdiffInt.}) respectively.\\
We are going to solve this model for $u_0=0.1$ and various $\kappa=0.02,~0.04,~0.1,~0.2$ and $~0.5$.
\subsection{Solving Volterra's Population Model by Rational Chebyshev Functions}
In the first step of our analysis, we apply $P_{N}$ operator on the function $y(t)$ as follows:
\begin{equation}
P_{N}y(t)=\sum_{k=0}^{N}a_{k}R_{k}(t)
\end{equation}
Then, we construct the residual function by substituting $y(t)$ by $P_{N}y(t)$ in
the volterra's population in Eq. (\ref{Main.voltra}):
\begin{align}
\nonumber Res(t)=&\kappa \frac{d^2}{dt^2}P_{N}y(t)-\frac{d}{dt}P_{N}y(t)+\left({\frac{d}{dt}P_{N}y(t)}\right)^2\\
&+\left( P_{N}y(t)\right)\left(\frac{d}{dt}P_{N}y(t)\right),
\end{align}
The equations for obtaining the coefficients $a_k$s come from equalizing $Res(t)$
to zero at rational Chebyshev-Gauss-Radau points plus two boundary conditions:
\begin{align}
\begin{cases}
Res(x_j)=0, \quad j=1,2,...,N-1,\\
P_{N}y(0)=0,\\
\frac{d}{dx}P_{N}y(t){\Big |}_{t=0}=u_{0}.
\end{cases}
\end{align}
Solving the set of equations we have the approximating function $P_{N}y(x)$.

\subsection{Solving Volterra's Population Model by Hermite functions}
As mentioned before, Volterra's Population Model is defined on the interval $(0,+\infty)$; but we know the properties of Hermite functions are derived in the infinite domain $(-\infty,+\infty)$.\\
Also we know approximations can be constructed for infinite, semi-infinite and finite intervals.
One of the approaches to construct approximations on the interval $(0,+\infty)$ which is used in this paper, is the use of mapping, that is a change of variable of the form
\begin{equation}\nonumber
\omega=\phi(z)=\ln(\sinh(kz)).
\end{equation}
where $k$ is a constant.\\
The basis functions on $(0,+\infty)$ are taken to be the transformed Hermite functions,
\begin{align}\nonumber
\widehat{H}_n(x)\equiv \widetilde{H}_n(x)\circ \phi(x)= \widetilde{H}_n(\phi(x)).
\end{align}
where $\widetilde{H}_n(x)\circ \phi(x)$ is defined by $\widetilde{H}_n(x)(\phi(x))$. The inverse map of $\omega=\phi(z)$ is
\begin{align}\label{inverseTransform}
z=\phi^{-1}(\omega)=\frac{1}{k}\ln (e^\omega +\sqrt{e^{2\omega}+1}).
\end{align}
Thus we may define the inverse images of the spaced nodes ${ {\{{x_j}}\}_{x_j=-\infty}^{x_j=+\infty} }$ as
\begin{align}\nonumber
\Gamma=\{\phi^{-1}(t): -\infty< t <+\infty\}=(0,+\infty)
\end{align}
and
\begin{align}\nonumber
\tilde{x}_j=\phi^{-1}(x_j)=e^{k{x_j}},\quad j=0,1,2,...
\end{align}
Let $w(x)$ denotes a non-negative, integrable, real-valued function over the interval $\Gamma$.
We define
\begin{align}\nonumber
L^2_w(\Gamma)=\{v:\Gamma\rightarrow \mathbb{R}\mid v \textrm{ is measurable and}\parallel v{\parallel}_w<\infty \}
\end{align}
where
\begin{align}\nonumber
\parallel v{\parallel}_w=\left(\int_{0}^\infty\mid v(x)\mid ^2w(x)\mathrm{d}x\right)^{\frac{1}{2}},
\end{align}
is the norm induced by the inner product of the space $L^2_w(\Gamma)$,
\begin{equation}\label{Eq.inner product definition}
<u,v>_w=\int_{0}^{\infty}u(x)v(x)w(x)\mathrm{d}x.
\end{equation}
Thus $\{\widehat{H}_n(x)\}_{n\in \mathbb{N}}$ denotes a system which is mutually orthogonal under (\ref{Eq.inner product definition}), i.e.,
\begin{align}\nonumber
< \widehat{H}_n(x),\widehat{H}_m(x)>_{w(x)}=\sqrt{\pi}\delta_{nm},
\end{align}
where $w(x)=\coth(x)$ and $\delta_{nm}$ is the Kronecker delta function. This system is complete in $L^2_w(\Gamma)$. For any function $f\in L^2_w(\Gamma)$ the following expansion holds

\begin{align}\nonumber
f(x)\cong \sum_{k=-N}^{+N}f_k \widehat{H}_k(x),
\end{align}
with
\begin{align}\nonumber
f_k=\frac{<f(x),\widehat{H}_k(x)>_{w(x)}}{\parallel \widehat{H}_k(x){\parallel}_{w(x)}^2}.
\end{align}
Now we can define an orthogonal projection based on transformed Hermite functions as below:\\
Let
\begin{align}\nonumber
\widehat{\mathcal{H}}_N=span\{\widehat{H}_0(x),\widehat{H}_1(x),...,\widehat{H}_n(x)\}
\end{align}
The $L^{2}{(\Gamma)}$-orthogonal projection $\hat{\xi}_N:L^{2}{(\Gamma)}\longrightarrow\widehat{\mathcal{H}}_N$ is a mapping in a way that for any $y\in L^{2}{(\Gamma)}$,
\begin{align}\nonumber
<\hat{\xi}_N{y}-y,\phi>=0, \quad \forall \phi\in \widehat{\mathcal{H}}_N
\end{align}
or equivalently,
\begin{align}\label{operatorHFT}
\hat{\xi}_N{y(x)}=\sum_{i=0}^{N}\hat{a}_{i}\widehat{H}_i(x).
\end{align}


Then to apply the HFC method to approximate the Volterra's Population Eq.(\ref{Main.voltra}) with initial conditions (\ref{voltraBoundry}), we use the operator introduced in Eq.(\ref{operatorHFT}) basically.

We note that the Hermite functions are not differentiable at the point $x=0$, therefore to approximate the solution of Eq. (\ref{Main.voltra}) with the initial conditions Eq. (\ref{voltraBoundry}) we construct a polynomial $p(t)$ that satisfy $y'(0)=u_0$ in Eq. (\ref{voltraBoundry}) and also multiply operator Eq. (\ref{operatorHFT}) by $t$ to make it differentiable at point $t=0$ in Eq. (\ref{voltraBoundry}).
This polynomial is given by
\begin{equation}\label{px}
p(t)=\lambda{t^2}+0.1t,
\end{equation}
where $\lambda$ is constant to be determined.\\

Therefore, the approximate solution of $y(t)$, in Eq. (\ref{Main.voltra}) with initial conditions Eq. (\ref{voltraBoundry}) is represented by
\begin{align}\label{operatorMHFT}
\widehat{\xi}_{N}y(t)= p(t)+t\hat{\xi}_Ny(t),
\end{align}

Now we construct the residual function by substituting $y(t)$ by $\widehat{\xi}_{N}y(t)$ in
the Volterra's population Eq. (\ref{Main.voltra}):
\begin{align}
\nonumber Res_{l}(t)=&\kappa \frac{d^2}{dx^2}\widehat{\xi}_{N}y(t/l)-\frac{d}{dx}\widehat{\xi}_{N}y(t/l)+\left({\frac{d}{dx}\widehat{\xi}_{N}y(t/l)}\right)^2\\
& +\left( \widehat{\xi}_{N}y(t/l)\right)\left(\frac{d}{dx}\widehat{\xi}_{N}y(t/l)\right),
\end{align}
where $l$ is a constant that is called domain scaling.

It has already been mentioned in \cite{Liu.Liu.Tang} that when using a spectral approach on the whole real line $\mathbb{R}$ one can possibly increase the accuracy of the computation by a suitable scaling of the underlying time variable $t$. For example, if $y$ denotes a solution of the ordinary differential equation, then the rescaled function is
$\tilde{y}(t)=y(t/l)$, where $l$ is constant. Domain scaling is used in several of the applications presented in the next section. For more detail we refer the reader to \cite{Tang.1993}.

The equations for obtaining the coefficients $a_i$s come from equalizing $Res(t)$
to zero at transformed Hermite-Gauss points:
\begin{align}
Res(x_j)=0, \quad j=0,\ldots,N+1,
\end{align}
where the $x_j$s are $N+2$ transformed Hermite-Gauss nodes by Eq. (\ref{inverseTransform}).
This generates a set of $N+2$ nonlinear equations that can be solved by Newton method for unknown coefficients $a_i$s and $\lambda$.

Table \ref{Tab.Results} Shows a comparison of methods in \cite{Parand.Hojati,Ramezani.Razzaghi}, and the present methods with the exact values
\begin{equation}\nonumber
u_{max}=1+\kappa\ln\left(\frac{\kappa}{1+\kappa-u_0}\right)
\end{equation}
evaluated in \cite{Tebeest}.

Figures \ref{Fig-RCC} and \ref{Fig-HFC} show the results of rational Chebyshev and Hermite functions collocation methods for $\kappa$= $0.02$, $0.04$, $0.1$, $0.2$, $0.5$. These figures show the rapid rise along the logistic curve followed by the slow exponential decay after reaching the maximum point and when $\kappa$ increases, the amplitude of $u(t)$ decreases whereas the exponential decay increases.

Figure \ref{Fig-HFC-RCC-k0.02} illustrates a comparison between the two presented methods for $\kappa=0.02$.

Logarithmic graphs of absolute coefficients $|a_i|$ of rational Chebyshev and Hermite functions in the approximate solutions for $\kappa=0.5$ are shown in Figures \ref{Fig-RCC-k05-Coeff} and \ref{Fig-HFC-k05-Coeff}, respectively. These graphs illustrate that both of methods have an appropriate convergence rate.

\section{Conclusions}
The aim of the this study is to develop an efficient and accurate numerical method based on orthogonal functions for solving the Volterra model for the population growth of a species in a closed system. The methods were used in a direct way on a semi-infinite domain without using linearization, perturbation or restrictive assumptions. In this paper, we have applied both the rational Chebyshev and the Hermite functions in solving nonlinear integro-differential equations and compared the results obtained by the two methods and others reported in \cite{Parand.Hojati,Ramezani.Razzaghi}. The study showed that rational Chebyshev collocation method is simple and easy to use. It also minimizes the computational results.
In total an important concern of spectral methods is the choice of basis functions; the basis functions have three properties: easy to computation, rapid convergence and completeness, which means that any solution can be represented. The stability and convergence of rational Chebyshev and Hermite functions approximations make this approach very attractive and contributing to the good agreement between the approximate and exact values for $u_{max}$ in the numerical example.

\section*{Acknowledgments}
The first author's research (K. Parand) was supported by a grant from Shahid Beheshti University.

\clearpage{}
\begin{table}
\caption{A comparison of methods in
\cite{Parand.Hojati,Ramezani.Razzaghi} and the present methods
with the exact  values for $u_{max}$}
\begin{tabular*}{\columnwidth}{@{\extracolsep{\fill}}*{8}{c}}
\hline
\multicolumn{2}{c}{}& \multicolumn{4}{c}{Present methods}& \multicolumn{2}{c}{Other methods}\\
\cline{3-6} \cline{7-8}
$\kappa$ & Exact $u_{max}$ & N &HFC & N & RCC & SDMM \cite{Parand.Hojati} & CSF with $N=100$ \cite{Ramezani.Razzaghi} \\
\hline
$0.02$ & $0.92342717$ & $20$& $0.92342704$ & $14$& $0.92342715$ & $0.92342714$ & $0.9234262$\\
$0.04$ & $0.87371998$ & $25$& $0.87371998$ & $14$& $0.87371998$ & $0.87381998$ & $0.8737192$\\
$0.1$  & $0.76974149$ & $20$& $0.76974149$ & $14$& $0.76974149$ & $0.76974140$  &$0.7697409$\\
$0.2$  & $0.65905038$ & $25$& $0.65905038$ & $11$& $0.65905038$ & $0.65905037$ & $0.6590497$\\
$0.5$  & $0.48519030$ & $30$& $0.48519030$ & $13$& $0.48519030$ & $0.48519029$ & $0.4851898$\\
\hline
\end{tabular*}
\label{Tab.Results}
\end{table}
\clearpage{}

\begin{figure}
\includegraphics[scale=0.4]{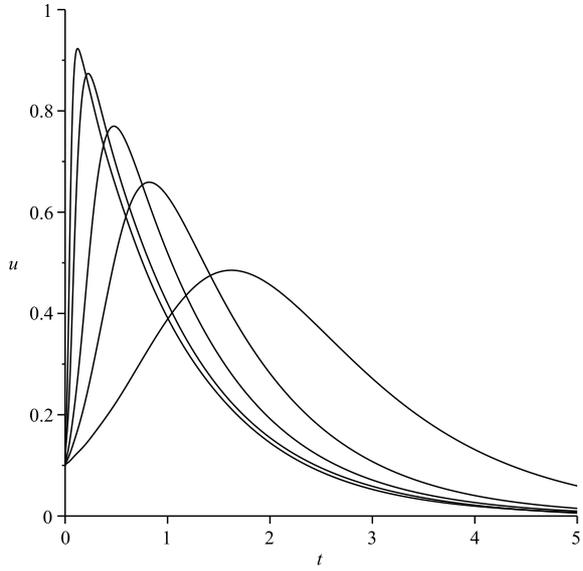}
\caption{The results of rational Chebyshev collocation method calculation for
$\kappa=0.02, 0.04, 0.1, 0.2, 0.5$, in the order of height}
\label{Fig-RCC}
\end{figure}

\begin{figure}
\includegraphics[scale=0.4]{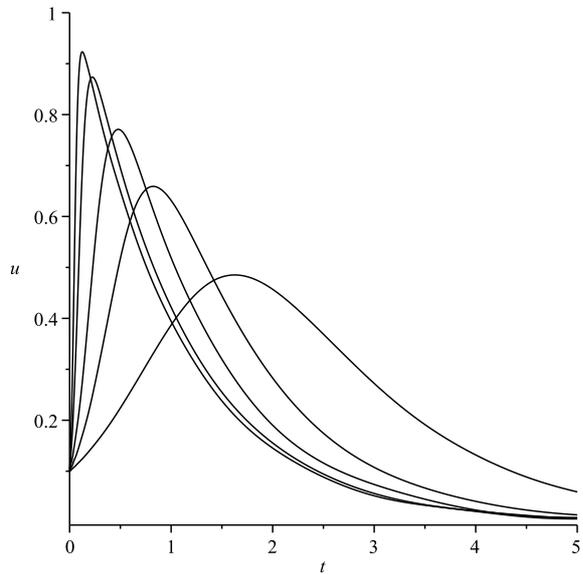}
\caption{The results of Hermite function collocation method calculation for
$\kappa=0.02, 0.04, 0.1, 0.2, 0.5$, in the order of height}
\label{Fig-HFC}
\end{figure}
\clearpage{}
\begin{figure}
\includegraphics[scale=0.4]{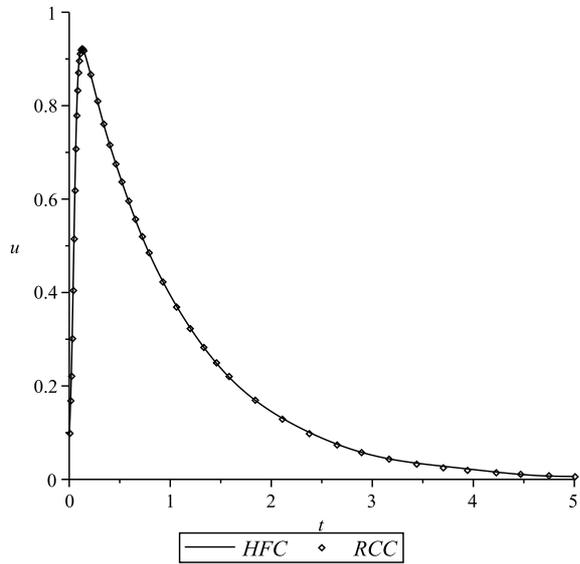}
\caption{The results of the comparison of rational Chebyshev and Hermite functions collocation method for $\kappa=0.02$}
\label{Fig-HFC-RCC-k0.02}
\end{figure}
\clearpage{}
\begin{figure}
\includegraphics[scale=0.4]{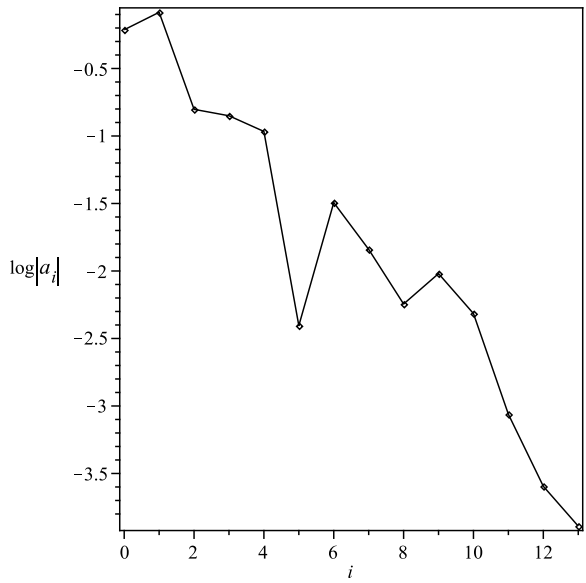}
\caption{The Logarithmic graph of absolute coefficients $|a_i|$ of rational Chebyshev functions in the approximate solution for $\kappa=0.5$}
\label{Fig-RCC-k05-Coeff}
\end{figure}

\begin{figure}
\includegraphics[scale=0.4]{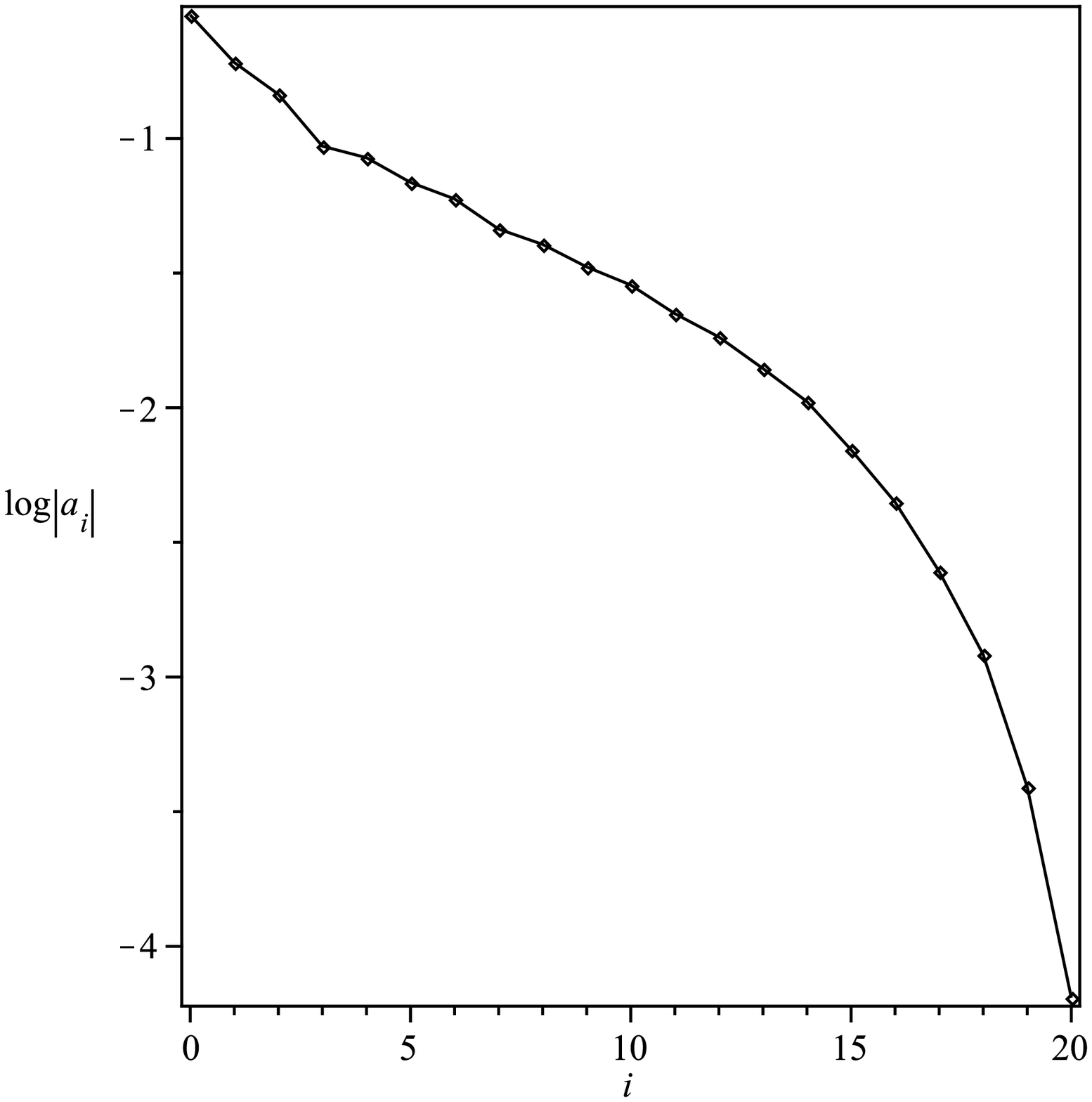}
\caption{The Logarithmic graph of absolute coefficients $|a_i|$ of Hermite functions in the approximate solution for $\kappa=0.5$}
\label{Fig-HFC-k05-Coeff}
\end{figure}


\begin{thebibliography}{99}
\bibitem{Coulaud}Coulaud~O, Funaro~D, Kavian~O. 
Laguerre spectral approximation of elliptic problems in exterior domains.
\emph{Computer Methods in Applied Mechanics and Engineering}.
1990; \textbf{80} (1-3):451--458.DOI: 10.1016/0045-7825(90)90050-V.

\bibitem{Funaro.Kavian} Funaro~D, Kavian~O. 
Approximation of some diffusion evolution equations in unbounded domains by Hermite functions.
\emph{Mathematics of Computation}.
1991; \textbf{57}:597--619.

\bibitem{Funaro.Appl. Numer. Math.1990} Funaro~D. 
Computational aspects of pseudospectral Laguerre approximations.
\emph{Applied Numerical Mathematics}.
1990; \textbf{6}(6):447--457. DOI: 10.1016/0168-9274(90)90003-X.

\bibitem{Guo.Math. Comp.1999}Guo~BY. 
Error estimation of Hermite spectral method for nonlinear partial differential equations.
\emph{Mathematics of Computation}.
1999; \textbf{68}(227):1067--1078. DOI: 10.1090/S0025-5718-99-01059-5.

\bibitem{Guo.num2000}Guo~BY, Shen~J. 
Laguerre-Galerkin method for nonlinear partial differential equations on a semi-infinite interval.
\emph{Numerische Mathematik}.
2000; \textbf{86}(4):635--654.DOI: 10.1007/s002110000168.

\bibitem{Maday}Maday~Y, Pernaud-Thomas~B, Vandeven~H. 
 Reappraisal of Laguerre type spectral methods.
\emph{La Recherche Aerospatiale}.
 1985; \textbf{6}:13--35.

 \bibitem{Shen}Shen~J. 
 Stable and efficient spectral methods in unbounded domains using Laguerre functions.
\emph{SIAM Journal on Numerical Analysis}.
 2000; \textbf{38}(4):1113--1133.DOI: 10.1137/S0036142999362936.

\bibitem{Siyyam}Siyyam~HI. 
Laguerre Tau methods for solving higher-order ordinary differential equations.
\emph{Journal of Computational Analysis and Applications}.
 2001; \textbf{3}(2):173--182.DOI: 10.1023/A:1010141309991.

\bibitem{Iranzo}Iranzo~V, Falques~A. 
Some spectral approximations for differential equations in unbounded domains.
\emph{Computer Methods in Applied Mechanics and Engineering},
  1992; \textbf{98}(1):105--126.DOI: 10.1016/0045-7825(92)90171-F.

\bibitem{Guo.J. Math. Anal. Appl.1998}Guo~BY. 
Gegenbauer approximation and its applications to differential equations on the whole line.
\emph{Journal of Mathematical Analysis and Applications}.
1998; \textbf{226}(1):180--206.

\bibitem{Guo.com2000}Guo~BY. 
Jacobi spectral approximation and its applications to differential equations on the half line.
\emph{Journal of Computational Mathematics}.
2000; \textbf{18}:95--112.

\bibitem{Guo.J. Math. Anal. Appl.2000}Guo~BY. 
Jacobi approximations in certain Hilbert spaces and their applications to singular differential equations.
\emph{Journal of Mathematical Analysis and Applications}.
2000; \textbf{243}(2):373--408.DOI: 10.1006/jmaa.1999.6677.

 \bibitem{BoydBook}Boyd~JP. 
\emph{Chebyshev and Fourier Spectrals Method}.
Dover Press: New York;
2001.

\bibitem{Christov.SIAM J. Appl. Math.1982}Christov~CI. 
A complete orthogonal system of functions in $L^2(-\infty,\infty)$ space.
\emph{SIAM Journal on Applied Mathematics}.
1982; \textbf{42}:1337--1344.

\bibitem{Boyd.J. Comput. Phys.1987(69)}Boyd~JP. 
Spectral methods using rational basis functions on an infinite interval.
\emph{Journal of Computational Physics}.
1987; \textbf{69}(1):112--142.DOI: 10.1016/0021-9991(87)90158-6.

\bibitem{Boyd1987}Boyd~JP. 
Orthogonal rational functions on a semi-infinite interval.
\emph{Journal of Computational Physics}.
1987; \textbf{70}(1):63--88.DOI: 10.1016/0021-9991(87)90002-7.

\bibitem{Guo.sci2000}Guo~BY, Shen~J, Wang~ZQ. 
A rational approximation and its applications to differential equations on the half line.
\emph{Journal of Scientific Computing}.
2000; \textbf{15}(2):117--147.DOI: 10.1023/A:1007698525506.

\bibitem{Boyd2003} Boyd~JP, Rangan~C, Bucksbaum~PH. 
Pseudospectral methods on a semi-infinite interval with application to the Hydrogen atom: a comparison of the mapped Fourier-sine method with Laguerre series and rational Chebyshev expansions.
\emph{Journal of Computational Physics}.
2003; \textbf{188}(1):56--74.DOI: 10.1016/S0021-9991(03)00127-X.

\bibitem{Parand.Appl. Math. Comput.2004} Parand~K, Razzaghi~M. 
Rational Chebyshev Tau method for solving Volterra's population model.
\emph{Applied Mathematics and Computation}.
 2004; \textbf{149}(3):893--900.DOI: 10.1016/j.amc.2003.09.006.

\bibitem{Parand.Int. J. Comput. Math.2004}Parand~K, Razzaghi~M. 
Rational Chebyshev Tau method for solving higher-order ordinary differential equations.
\emph{International Journal of Computer Mathematics}.
2004; \textbf{81}(1):73--80.DOI: 10.1080/00207160310001614981.

\bibitem{Parand.Phys. Scripta2004}Parand~K, Razzaghi~M. 
Rational Legendre approximation for solving some physical problems on semi-infinite intervals.
\emph{Physica Scripta}.
2004; \textbf{69}:353--357.DOI: 10.1238/Physica.Regular.069a00353.

\bibitem{Parand.Shahini.Phys.let.A} Parand~K, Shahini~M. 
Rational Chebyshev pseudospectral approach for solving Thomas-Fermi equation.
\emph{Physics Letters A}.
2009; \textbf{373}:210--213.DOI: 10.1016/j.physleta.2008.10.044.

\bibitem{Parand.CAM} Parand~K, Taghavi~A. 
Rational scaled generalized Laguerre function collocation method for solving the Blasius equation.
\emph{Journal of Computational and Applied Mathematics}.
2009; \textbf{233}(4):980--989. DOI: 10.1016/j.cam.2009.08.106.

\bibitem{Parand.JCP} Parand~K, Shahini~M, Dehghan~M. 
Rational Legendre pseudospectral approach for solving nonlinear differential equations of Lane-Emden type.
\emph{Journal of Computational Physics}.
2009; \textbf{228}(23):8830--8840.DOI: 10.1016/j.jcp.2009.08.029.


\bibitem{Parand.Dehghan.Rezaei.CPC} Parand~K, Dehghan~M, Rezaei~AR, Ghaderi~SM.
An approximational algorithm for the solution of the nonlinear Lane-Emden type equations arising in astrophysics using Hermite functions collocation method \emph{Computer Physics Communications}.
2010; DOI:10.1016/j.cpc.2010.02.018.

\bibitem{Scudo}Scudo~FM. 
Vito Volterra and theoretical ecology.
\emph{Theoretical Population Biology}.
 1971; \textbf{2}(1):1--23.DOI: 10.1016/0040-5809(71)90002-5.

\bibitem{Small}Small~RD. 
Population growth in a closed system.
\emph{SIAM review}.
 1983; \textbf{25}(1):93--95.

\bibitem{Tebeest} TeBeest~KG. 
Numerical and analytical solutions of Volterra's population model.
\emph{SIAM review}.
 1997; \textbf{39}(3):484--493.

\bibitem{Al-Khaled}Al-Khaled~K. 
Numerical approximations for population growth models.
\emph{Applied Mathematics and Computation}.
 2005; \textbf{160}(3):865--873.DOI: 10.1016/j.amc.2003.12.005.

\bibitem{Wazwaz}Wazwaz~AM. 
Analytical approximations and Pad\'{e} approximants for Volterra's population model.
\emph{Applied Mathematics and Computation}.
1999; \textbf{100}(1):13--25.DOI: 10.1016/S0096-3003(98)00018-6.

\bibitem{Parand.Hojati} Parand~K, Hojjati~G. 
Solving Volterra's Population Model using new Second Derivative Multistep Methods.
 \emph{American Journal of Applied Sciences}.
 2008; \textbf{5}(8) 1019--1022.DOI: 10.3844/ajassp.2008.1019.1022.

\bibitem{Ramezani.Razzaghi} Ramezani~M, Razzaghi~M, Dehghan~M. 
Composite spectral functions for solving Volterra's population model.
\emph{Chaos, Solitons \& Fractals}.
2007; \textbf{34}(2):588-–593.DOI: 10.1016/j.chaos.2006.03.067.

\bibitem{Marzban.Hoseini}Marzban~HR, Hoseini~SM, Razzaghi~M. 
Solution of Volterra's population model via block-pulse functions and Lagrange-interpolating polynomials.
\emph{Mathematical Methods in the Applied Sciences}.
2009; \textbf{32}:127--134. DOI: 10.1002/mma.1028.

\bibitem{Momani.Qaralleh}Momani~S, Qaralleh~R. 
Numerical approximations and Pad\'{e} approximants for a fractional population growth model.
\emph{Applied Mathematical Modelling}.
2007; \textbf{31}:1907--1914.DOI: 10.1016/j.apm.2006.06.015.

\bibitem{Xu}Xu~H. 
Analytical approximations for a population growth model with fractional order.
\emph{Communications in Nonlinear Science and Numerical Simulation}.
2009; \textbf{14}:1978--1983.DOI: 10.1016/j.cnsns.2008.07.006.

\bibitem{Guo.Shen.Wang2002}Guo~BY, Shen~J, Wang~ZQ. 
Chebyshev rational spectral and pseudospectral methods on a semi-infinite interval.
\emph{International journal for numerical methods in engineering}.
2002; \textbf{53}(1):65--84.DOI: 10.1002/nme.392.

\bibitem{Boyd1982}Boyd~JP. 
The Optimzation of Convergence for Chebyshev Polynomial Methods in an Unbounded Domain.
\emph{Journal of Computational Physics}.
1982; \textbf{45}:43--79.DOI: 10.1016/0021-9991(82)90102-4.

\bibitem{ShenWang2008} Shen~J, Wang~L-L. 
Some Recent Advances on Spectral Methods for Unbounded Domains.
\emph{Communications in computational physics}.
2009; 5 (2-4):195--241.

\bibitem{ShenTangHighOrder}Shen~J, Tang~T.  
\emph{High Order Numerical Methods and Algorithms}.
Chinese Science Press, Beijing; 2005.

\bibitem{ShenTangWang}Shen~J, Tang~T, Wang~L-L.  
\emph{Spectral Methods Algorithms, Analyses and Applications}.
Springer, First edition, 2010.

\bibitem{GuoShenXu2003} Guo~BY, Shen~J, Xu~C-l. 
Spectral and pseudospectral approximations using Hermite functions: application to the Dirac equation.
\emph{Advances in Computational Mathematics}.
2003; \textbf{19}(1-3):35--55. DOI: 10.1023/A:1022892132249.

\bibitem{Liu.Liu.Tang}Liu~Y, Liu~L, Tang~T.  
The numerical computation of connecting orbits in dynamical systems: a rational spectral approach.
\emph{Journal of Computational Physics}.
1994; \textbf{111}(2):373--380.DOI: 10.1006/jcph.1994.1070.

\bibitem{Tang.1993}Tang~T. 
The Hermite spectral method for Gaussian-type functions.
\emph{SIAM Journal on Scientific Computing}.DOI: 10.1137/0914038.
1993; \textbf{14}(3):594--606.


\end{thebibliography}
\end{document}